\documentclass[sigconf]{acmart} 

\AtBeginDocument{%
  \providecommand\BibTeX{{%
    Bib\TeX}}}

\setcopyright{acmcopyright}
\copyrightyear{2023}
\acmYear{2023}
\acmDOI{XXXXXXX.XXXXXXX}

\acmConference[Conference acronym 'XX]{Make sure to enter the correct
  conference title from your rights confirmation email}{June 03--05,
  2023}{Woodstock, NY}
\acmPrice{15.00}
\acmISBN{978-1-4503-XXXX-X/18/06}

\usepackage{graphicx}
\AtBeginDocument{%
  \providecommand\BibTeX{{%
    \normalfont B\kern-0.5em{\scshape i\kern-0.25em b}\kern-0.8em\TeX}}}

\usepackage{graphicx}
\usepackage{wrapfig}
\usepackage{booktabs} 

\usepackage{url}
\usepackage{algorithm}
\usepackage[noend]{algpseudocode}
\usepackage{amsmath}
\usepackage{subcaption}
\usepackage{graphicx}
\usepackage{comment}
\usepackage{thmtools, thm-restate}
\usepackage{multirow}
\usepackage{cite}

\usepackage[graphicx]{realboxes}

\usepackage{xcolor} 
\definecolor{ocre}{RGB}{243,102,25} 

\usepackage{avant} 

\usepackage{microtype} 
\usepackage[utf8x]{inputenc} 

\usepackage{tikz}
\usetikzlibrary{shapes}
\RequirePackage{ifthen}
\RequirePackage{amsmath}
\RequirePackage{xspace}
\RequirePackage{graphics}
\usepackage{color}
\usepackage{keyval}
\usepackage{xspace}
\usepackage{mathrsfs}

\newcommand{\reals}{{\mathbb{R}}} 

\newcommand{\nnreals}{{\mathbb{R}_{\geq 0}}}
 
\newcommand{\dom}{\relax\ifmmode {\mathit{dom}} \else ${\sf dom}$\fi}

\newcommand{\rec}{\mathit{rec}}
\newcommand{\inv}{\mathit{inv}}

\usepackage{setspace}
\usepackage{listings}

\newcounter{theorems}
\newcounter{corollaries}

\newcounter{lemmas}
\newcounter{remarks}
 \newtheorem{corollary}[corollaries]{Corollary}
 \newtheorem{theorem}[theorems]{Theorem}
\newtheorem{lemma}[lemmas]{Lemma}
\newcounter{assumptions}
\newtheorem{assumption}[assumptions]{Assumption} 
\newtheorem{remark}[remarks]{Remark}
\newcounter{definitions}
\newtheorem{definition}[definitions]{Definition}

\usepackage[normalem]{ulem}

\usepackage{placeins}
\makeatletter
\newcommand\nocaption{%
    \renewcommand\p@subfigure{}
    \renewcommand\thesubfigure{\thefigure\alph{subfigure}}
}
\makeatother

\DeclareMathDelimiter{(}{\mathopen} {operators}{"28}{largesymbols}{"00}
\DeclareMathDelimiter{)}{\mathclose}{operators}{"29}{largesymbols}{"01}

\begin{document}

\title{Recurrence of Nonlinear Control Systems:\\ Entropy and Bit Rates}

\author{Hussein Sibai}
\email{sibai@wustl.edu}
\affiliation{%
	\institution{Washington University in St. Louis }
	\country{USA}
}

\author{Enrique Mallada }
\email{ mallada@jhu.edu}
\affiliation{%
	\institution{Johns Hopkins University }
	\country{USA}
}



\begin{abstract}
In this paper, we introduce the notion of recurrence entropy in the context of nonlinear control systems.
A set is said to be ($\tau$-)recurrent if every trajectory that starts in the set returns to it (within at most $\tau$ units of time).
Recurrence entropy quantifies the complexity of making a set $\tau$-recurrent measured by the average rate of growth, as time increases, of the number of control signals required to achieve this goal.
Our analysis reveals that, compared to invariance, recurrence is quantitatively less complex, meaning that the recurrence entropy of a set is no larger than, and often strictly smaller than, the invariance entropy. Our results further offer insights into the minimum data rate required for achieving recurrence. We also present an algorithm for achieving recurrence asymptotically.
\end{abstract}




\maketitle


	\section{Introduction}\label{sec:intro}
	
	The topological entropy of a dynamical system is a fundamental property, an invariant~\citep{katok1995introduction}, that describes the rate of the exponential growth of the number of trajectories that are distinguishable with arbitrarily small but finite accuracy.
	Originally proposed by Adler, Konheim, and McAndrew~\citep{adler1965topological}, and shortly after reformulated in the form described above by Bowen~\citep{bowen1971entropy,bowen1971periodic}, it provides a quantitative measure of complexity by capturing how the uncertainty around the system state grows as time evolves.
	As a result, topological entropy is closely related to information-theoretic notions, such as the average rate of information gathering about the system state above which one can distinguish its trajectories with arbitrary accuracy~\citep{liberzon2016entropy}. 
	
	In control theory, wherein one uses the system's state information to perform a task, several notions of entropy have been proposed in the literature, characterizing the complexity of and the minimal data rates necessary for performing a certain control task. Examples of this include estimation entropy~\citep{savkin2006analysis,liberzon2016entropy,sibai2017optimal,sibai_hscc_2018,sibai_entropy_tac,kawan2018optimal}, restoration entropy~\citep{matveev2016observation,matveev2019observation}, stabilization entropy~\citep{delchamps1990stabilizing,nair2004topological,colonius2012minimal}, among others.
	One notion of entropy particularly instrumental in control is the invariance entropy~\citep{colonius-kawan-2009,colonius2011invariance,colonius2013note,tomar2022numerical}, which aims to capture the growth rate of the number of distinct control signals necessary to render a certain set invariant. 
	
	Invariance holds a prominent role in control theory. It is, for instance, a core notion in the development of the Lyapunov theory~\citep{Khalil2002}.  By trapping trajectories on sub-level sets of a Lyapunov function, one can guarantee boundedness and completeness of trajectories, stability, and even asymptotic or exponential stability via a gradual reduction of the value of the function. Invariant sets can also be used to estimate regions of attractions of an asymptotically stable equilibrium~\citep{genesio1985estimation}. However, due to intrinsic coupling between the dynamics of the system and the geometry of the set, finding invariant sets and, by extension, Lyapunov functions is often difficult. Furthermore, in the context of controlled systems, it is not always possible to make a given set (controlled) invariant.


	In this work, motivated by recent literature aimed at using the notion of recurrent sets as functional substitutes for invariant sets in control theory~\citep{sbm2022cdc,sspm2023cdc}, we introduce the notion of recurrence entropy for nonlinear control systems. A set is said to be ($\tau$-)recurrent if every trajectory that starts in the set returns to it (within at most $\tau$ units of time). Our analysis shows that recurrence, as a control task, is quantitatively less complex than invariance from the point of view that for a given set and dynamical system, the recurrence entropy is no larger than the invariance entropy.
	Furthermore, we provide upper and lower bounds for the recurrence entropy, a characterization of recurrence entropy as the minimum data rate necessary to render a set recurrent, as well as an accompanying algorithm that achieves this task asymptotically with a bit rate equal to the recurrence entropy upper bound that we derive plus a linear term equal to the product of the system dimension and the desired rate of exponential convergence toward a recurrent trajectory.
	
	\emph{Related Work:} Our work is closely related to the literature of invariance entropy~\citep{colonius-kawan-2009,colonius2011invariance,colonius2013note}. Naturally, since every invariant set is (trivially) $\tau$-recurrent, for all $\tau\geq0$, the results presented therein apply for $\tau=0$. 
 Our work also relates to that of Tomar et al. \citep{tomar2023towards}. That work relates the minimal bit rates needed to enforce a regular safety property for a discrete-time dynamical system to the invariance entropy of a new system that combines the automaton defining the property and the original system. Particularly, $\tau$-recurrence can  
	be thought of as 
	regular safety property, but as we define it here, it is for continuous-time dynamical systems. Relating  our results with \citep{tomar2023towards} would be an interesting future direction.
	It would also be interesting to design numerical methods such as those proposed in \citep{tomar2022numerical} for invariance entropy to estimate recurrence entropy.

	\emph{Organization of the Paper:}
	The rest of the paper is organized as follows. In Section \ref{sec:prelim}, we provide preliminary definitions regarding the system to be considered, as well as the notion of invariance entropy. In Section \ref{sec:recurrence entropy}, we formally introduce the notion of recurrence to be studied in the paper, i.e., $\tau$-recurrence (c.f. Definition \ref{def:recurrent_trajectories}), as well as the associated notion of entropy. We then introduce, in Section \ref{sec:containment lemma}, a fundamental result that allows us to bound the distance from a set that recurrent trajectories can travel if they are required to come back to the set within $\tau$ units of time. A comparison between recurrence and invariance entropy is performed in Section \ref{sec:relation invariance entropy}, upper and lower bounds for recurrence entropy are provided in Section \ref{sec:bounds}, and a relationship between entropy and data rates is formally established in Section \ref{sec:data rates}. We finalize by introducing an algorithm that can make trajectories (asymptotically) $\tau$-recurrent in Section \ref{sec:algorithm} and giving final remarks in Section \ref{sec:conclusions}.
	
	
	

	\section{Preliminaries}\label{sec:prelim}
	
	\emph{Notation:} 
	Given a set $S$, $\text{cl}(S)$ denotes its closure, and $\mathrm{co}(S)$ its convex hull. We denote by $\|\cdot\|$ an arbitrary norm over $\reals^n$, unless otherwise specified. If $S$ is finite, $|S|$ denotes its cardinality. If $N \in \mathbb{N}$, we denote by $[N]$ the set of all non-negative integers less than $N$. We denote the closed $\infty$-norm ball, or hyperrectangle, centered at $x \in \mathbb{R}^n$ with radius $r \in \reals^{\geq 0}$ by $B(x, r)$. Given a compact set $S$, a $\delta$-cover of $S$ is a set of balls of radius $\delta$ whose union contains $S$. We abuse notation and call the set containing the centers of the balls the cover instead of the balls themselves. 
	We denote by $\mathit{grid}(S, \delta)$ the $\delta$-cover of $S$ that is constructed with the centers of $\infty$-norm balls which are $2\delta$ apart and located on axis-parallel lines. We also call it a $\delta$-grid of $S$. We assume that the logarithm function $\log$ is that of base 2 throughout the paper. 
	Fix an $\varepsilon>0$ and a compact set $Q \subset \reals^n$. We define $N_{\varepsilon} (Q) := \{ y \in \reals^n | \exists x \in Q, \|x-y\| \leq \varepsilon \}$. We also define $\lambda_n(Q)$ to be the Lebesgue measure of $Q$. Consider a function with $N$ arguments. If we replace its $i^{\mathit{th}}$ argument with ``$\cdot$'' (i.e., dot) and its other arguments with constants, we mean the projection of that function to the one-dimensional domain of the $i^{\mathit{th}}$ argument with the other ones fixed to the specified constants. If we replace an argument with a set in its domain, we mean the function defined only over that set in the domain.
	
	\subsection{System description} 
	In this paper, we consider control systems that are defined as follows.
	\begin{definition}
		Consider a nonlinear control system of the form:
		\begin{align}
			\label{eq:system}
			\dot{x}(t) = f(x(t),u(t)),
		\end{align}
		where $x(t) \in \reals^n$, $u \in \mathcal{U}$, with $\mathcal{U}$ being a set of piece-wise continuous functions mapping $\nnreals$ to {a compact set $U\subset \reals^m$}, and the map $f$ is locally Lipschitz. 
		We often abuse notation and interchangeably refer to the function $u\in\mathcal{U}$ as well as an input vector $u\in U$. 
	\end{definition}

	\subsection{Invariance entropy}
	
	In this section, we recall the definition of invariance entropy of system~(\ref{eq:system}) from \citep{colonius-kawan-2009}. It requires the definitions of controlled invariant sets, invariant trajectories, and invariance spanning sets. 
	
	\begin{definition}[Controlled invariant sets~\citep{colonius-kawan-2009}]
		A set $Q \subseteq \reals^m$ is {\em controlled invariant} for system~(\ref{eq:system}) if $\forall x \in Q$, $\exists u \in \mathcal{U}$ such that for any $t \geq 0$, $\xi(x,u,t) \in Q$. 
	\end{definition}
	We call controlled invariant sets invariant from hereafter for brevity.
	
	\begin{definition}[$(T,\varepsilon, Q)$-invariant trajectories~\citep{colonius-kawan-2009}]
		Fix any $\varepsilon \geq 0$, $T \geq 0$, compact set $Q \subset \reals^n$, $x \in Q$, and $u \in \mathcal{U}$. 
		The trajectory $\xi(x,u,\cdot)$ of system~(\ref{eq:system}) is $(T,\varepsilon, Q)$-invariant, if for every $t\in [0,T]$, $\xi(x,u,t) \in N_{\varepsilon}(Q)$.  If the condition is $\xi(x,u,t) \in Q$  instead, we say that $\xi$ is $(T, Q)$-invariant.
	\end{definition}
	
	Fix two non-empty sets $K \subseteq Q \subset \reals^n$, an $\varepsilon \geq 0$, and a $T\geq 0$. 
	A set $S\subseteq \mathcal{U}$ is called an {\em invariance} $(T,\varepsilon, K, Q)$-{\em spanning} set if for any $x \in K $, there exists a $u \in S$, such that $\xi(x,u,[0,T])$ is $(T,\varepsilon,Q)$-invariant. Let $r_{\inv}(T,\varepsilon, K, Q)$ be the minimal cardinality of such a set if it exists, and be equal to infinity otherwise. The {\em invariance} entropy of system~(\ref{eq:system}) is defined in \citep{colonius-kawan-2009} as follows:
	\begin{align}
		h_{\inv}(K, Q) &:= \lim_{\varepsilon\searrow 0}\limsup_{T\rightarrow \infty} \frac{1}{T} \log r_{\inv}(T,\varepsilon, K, Q). \label{def:inv_entropy} 
	\end{align}
	
	If the trajectories are required to be
	$(T,Q)$-invariant in 
	(\ref{def:inv_entropy}), then the minimal cardinality of the corresponding invariance spanning set is denoted in \citep{colonius-kawan-2009} by
	$r_{\inv}^*(T,K, Q)$.
	If substituted in 
	(\ref{def:inv_entropy}), the resulting entropy notion 
	$h_{\inv}^*(K, Q)$ is called the 
	{\em strict invariance} entropy of system~(\ref{eq:system}).
	When $K$ is equal to $Q$, we drop the $K$ argument in the definitions above.

	\section{$\tau$-Recurrence Entropy}\label{sec:recurrence entropy}
	
	In this section, we define the main concept that we contribute in this paper: $\tau$-recurrence entropy. Before being able to define it, we need to define controlled $\tau$-recurrent sets, recurrent trajectories, and recurrence spanning sets, in parallel with the definitions preceding the definition of invariance entropy in the previous section.  
	
	\subsection{Recurrence spanning sets and entropy}

	
In the following definition, we introduce controlled $\tau$-recurrent sets as compact subsets of the state space of system~(\ref{eq:system}) which satisfy the following condition:  for each state in such a set, there exists a control signal that drives the system to have a trajectory that visits the set at least once within each time interval of size $\tau$. This concept generalizes for non-autonomous systems the notion of $\tau$-recurrent sets, first introduced in \citep{sbm2022cdc}.
We then define the concept of $(T, \varepsilon, \tau, Q)$-recurrent trajectories, which are ones that return to $N_{\varepsilon}(Q)$ at least once within each time interval of size $\tau$ in the interval $[0,T]$.  


\begin{definition}[Controlled $\tau$-recurrent sets]
	A set $Q \subseteq \reals^m$ is {\em controlled $\tau$-recurrent} for system~(\ref{eq:system}), for some finite $\tau \in \reals^{\geq 0}$, 
	if for every $x \in Q$, there exists a $u \in \mathcal{U}$ such that for any $t \geq 0$, there exists a $t' \in [t, t+\tau]$ such that $\xi(x,u,t') \in Q$. 
\end{definition}
We call controlled $\tau$-recurrent sets $\tau$-recurrent from hereafter for brevity.

\begin{definition}[$(T,\varepsilon, \tau,Q)$-recurrent trajectories]
	\label{def:recurrent_trajectories}
	Fix any $\tau \geq 0$, $\varepsilon \geq 0$, $T \geq \tau$, compact set $Q \subset \reals^n$, $x \in Q$,  and $u \in \mathcal{U}$. The trajectory $\xi(x,u,\cdot)$ of system~(\ref{eq:system}) is $(T, \varepsilon, \tau,Q)$-recurrent, if for every $t\in [0,T - \tau]$, there exists a $t' \in [t, t+\tau]$ such that $\xi(x,u,t') \in N_{\varepsilon}(Q)$.
\end{definition}

For simplicity of notation, if $\varepsilon = 0$, we drop the $\varepsilon$ argument. Similarly, if 
$T =\infty$, we drop 
the $T$ argument. 
We will also use Definition~\ref{def:recurrent_trajectories} for functions of time, that are not necessarily trajectories of system~(\ref{eq:system}) or even continuous, but just piece-wise continuous.


The final definition before that of $\tau$-recurrence entropy is that of spanning sets. They are sets of control signals which are sufficient to make any trajectory starting from a $\tau$-recurrent set $\tau$-recurrent. 

Fix a $\tau \in \reals^{\geq 0}$, a compact $\tau$-recurrent set $Q \subset \reals^n$, an $\varepsilon \geq 0$, and a $T \geq 0$. 
A set $S\subseteq \mathcal{U}$ is called a {\em recurrence} $(T,\varepsilon,\tau, Q)$-{\em spanning} set if for any $x \in Q$, there exists a $u \in S$ such that $\xi(x,u,[0,T])$ is $(T,\varepsilon,\tau,Q)$-recurrent. Let $r_{\rec}(T,\varepsilon,\tau, Q)$ be the minimal cardinality of such a set  if it exists, and be equal to infinity otherwise. We define the $\tau$-{\em recurrence} entropy of system~(\ref{eq:system}) as follows:
\begin{align}
	h_{\rec}(\tau,Q) &:= \lim_{\varepsilon\searrow 0}\limsup_{T\rightarrow \infty} \frac{1}{T} \log r_{\rec}(T,\varepsilon,\tau, Q). \label{def:tau_rec_entropy} 
\end{align}

If we require the trajectories to be
$(T,\tau,Q)$-recurrent in 
(\ref{def:tau_rec_entropy}), then we denote the minimal cardinality of the corresponding spanning set
$r_{\rec}^*(T,\tau,Q)$.
If substituted  in 
(\ref{def:tau_rec_entropy}), we call the resulting entropy notion 
$h_{\rec}^*(\tau,Q)$ the 
{\em strict $\tau$-recurrence} entropy of system~(\ref{eq:system}). 



\section{Containment Lemma}\label{sec:containment lemma}

In this section, we show how trajectories that are $(\tau, Q)$-recurrent cannot depart arbitrarily from $Q$.
The following assumption is instrumental in achieving this goal.
\begin{assumption}[$\tau$-completeness]\label{ass:tau-completeness}
	For any $x\in Q$ and $u\in \mathcal{U}$, the trajectory $\xi(x,u,\cdot)$ is  defined for all $t\in[0,\tau]$ and is continuous in its first argument.
\end{assumption}

An immediate consequence of Assumption \ref{ass:tau-completeness} is that for any $u\in\mathcal{U}$, the closure reachable set  $R(Q,u, \tau) := {\cup_{t\in[0,\tau],x\in Q} \xi(x,u,t)}$ of system~(\ref{eq:system}), i.e., $\text{cl}({R(Q,u,\tau)})$, is compact.
Moreover, it follows from Proposition 5.2 \citep{lin1996smooth}, that under Assumption \ref{ass:tau-completeness}, the set $${R(Q,\tau)}:={\bigcup_{u\in \mathcal{U}}R(Q,u,\tau)}$$ is bounded.
The set ${R(Q,\tau)}$ contains all states visited by trajectories starting from \emph{some} initial state $x\in Q$ and following \emph{some} control $u\in\mathcal{U}$. While such a set is, indeed, bounded, it may be quite big, as not all control inputs are meant to make trajectories recurrent. We will therefore consider the subset $\mathcal{U}_r\subseteq \mathcal{U}$ containing all control inputs $u\in\mathcal{U}$ such that there exists some $x\in Q$, making the trajectory $\xi(x,u,\cdot)$ $(\tau,Q)$-recurrent. 

A similar reasoning as before, using the fact that $\mathcal{U}_r\subseteq\mathcal{U}$, leads to fact that the set
\begin{equation}
\label{eq:R_r}
	R_r(Q,\tau) :=\bigcup_{u\in\mathcal{U}_r}R(Q,u,\tau) \subseteq R(Q,\tau)\,,
\end{equation}
and is therefore bounded. For the purpose of estimating how far out $(\tau,Q)$-recurrent trajectories can reach, we define
\begin{equation}\label{eq:Ltau}
	L_\tau =  \max_{x_1,x_2\in \text{cl}\left({\mathrm{co}(R_r(Q,\tau))}\right),u\in U}\!\!\!\!\!\!\frac{\|f(x_1,u) - f(x_2,u)\|}{\|x_1 - x_2\|} \,<\infty.    
\end{equation}
Note that $L_\tau$ is an upper bound of the Lipschitz constant of the vector field along any $(\tau, Q)$-recurrent trajectory.




The following lemma, which is a generalization of Lemma 2 in \citep{sspm2023cdc}, allows us to obtain an estimate of how far trajectories can go outside a compact set $Q$ within $\tau$ seconds.
\begin{lemma}[Containment Lemma]\label{lem:containment}
	Consider a compact controlled $\tau$-recurrent set $Q$. 
	Then, given any $x\in{Q}$, and $u\in\mathcal{U}_r$ such that $\xi(x,u,\cdot)$ is $(\tau,Q)$-recurrent,  
	the following holds:
	\begin{equation}\label{eq:containment-Q}
		\sup_{t\in\mathbb{R}_{\geq0}} d(\xi(x,u,t),{Q}) \leq F_{Q}\tau e^{L_\tau \tau},
	\end{equation}
	where 
	$d(y,{Q}):=\min_{x\in {Q}} \|y-x\|$,  $L_\tau$ is given in \eqref{eq:Ltau}, and $$F_{Q}:=\sup_{x\in Q,\, u\in U} \|f(x, u)\|<\infty.$$
\end{lemma}
\begin{proof}
	As mentioned before, the proof of this lemma is akin to \citep{sspm2023cdc}, Lemma 2.
	Given $x\in Q$ and the corresponding $u\in \mathcal{U}_r$ that makes $\xi(x,u,\cdot)$ $(\tau,Q)$-recurrent, let $t_1>0$ be the first time the trajectory leaves $Q$, i.e., such that $\xi(x,u,t)\in Q$, for $t\leq t_1$, and for all $\delta>0$ sufficiently small $\xi(x,u,t+\delta)\not\in Q$. Without loss of generality, we assume $t_1<+\infty$.  It then follows from the $(\tau,Q)$-recurrent that for all $t\in [0,t_1+\tau]$,  $\xi(x,u,t)$ can only be outside $Q$ for at most $\tau$ seconds. Using now the short notation $x(t)=\xi(x,u,t)$ we have
	\begin{align*}
		a(t)&:=d(x(t),Q)\leq \|x(t)-x\|\\
		&= \left\|\int_0^t f(x(s),u(s)) ds\right\|\\
		&\leq \int_0^t \|f(x(s),u(s))-f(\Pi_Q[x(s)], u(s)\|\\ & \quad+ \|f(\Pi_Q[x(s)], u(s))\| ds\\
		&\leq \left(\int_{t_1}^t a(s) L_\tau ds\right)_+  + F_Q(t-t_1)_+.
	\end{align*}
	It follows the from Grönwall's inequality (c.f Lemma 2.1 in \citep{Khalil2002}, with $\lambda = F_Q (t-t_1)_+, \mu = L, y(t) = a(t)$)
	that $\forall t\in[0,t_1+\tau]$, 
	\[
	a(t)=d(\xi(x,u,t),Q)\leq F_Q(t-t_1)_+ e^{L_\tau^*(t-t_1)_+}\leq F_Q\tau e^{L_\tau^*\tau}.
	\]
	Finally, by repeating the same argument every  time $\xi(x,u,t)$ leaves $Q$, the result follows.
\end{proof}





\section{Relation between recurrence, $\tau$-recurrence, and invariance entropy}\label{sec:relation invariance entropy}

In this section, we show different relations between recurrence and invariance entropy of system~(\ref{eq:system}). In Theorem~\ref{thm:recurrence_lessthan_invariance}, we show that $\tau$-recurrence entropy is both lower and upper bounded by invariance entropy with different initial and invariant sets. That results in a corollary showing that as $\tau$ approaches zero, $\tau$-recurrence entropy approaches invariance entropy, which is in agreement with the intuition that $\tau$-recurrence with $\tau = 0$ is invariance. In Theorem~\ref{thm:recurrence_lessthan_tau_recurrence}, we show that $\tau$-recurrence entropy is less than $\tau'$-recurrence entropy if $\tau' \geq \tau$. That is in agreement with the intuition that faster recurrence to $Q$ requires more information about the state. 



\begin{theorem}
	\label{thm:recurrence_lessthan_invariance}
	For any $Q \subseteq \reals^n$ that is controlled invariant and $\tau \geq 0$, $h_{inv}(Q, N_{\delta_\tau}(Q)) \leq h_{\rec}(\tau, Q) \leq h_{\mathit{inv}}(Q)$ and $h_{inv}^*(Q, N_{\delta_\tau}(Q)) \leq h_{\rec}^*(\tau, Q) \leq h_{\mathit{inv}}^*(Q)$, where $\delta_\tau$ is the right-hand-side of the containment lemma.
\end{theorem}
\begin{proof}
The first inequality follows from the containment lemma that shows that any recurrence $(T,\varepsilon, \tau, Q)$-spanning (resp. $(T, \tau, Q)$-spanning) set is an invariance $(T, \varepsilon, Q, N_{\delta_\tau}(Q))$-spanning (resp. $(T, Q,$  $N_{\delta_\tau}(Q))$-spanning) set. 
	The second inequality follows from the observation that any invariance $(T,\varepsilon, Q)$-spanning (resp. $(T, Q)$-spanning) set is a recurrence $(T,\varepsilon, \tau,Q)$-spanning (resp. $(T,\tau,Q)$-spanning) set as well, for any $\tau \geq 0$ and $\varepsilon \geq 0$. 
\end{proof}

\begin{corollary}
	As $\tau\rightarrow 0$, $\tau$-recurrence entropy becomes equal to invariance entropy, i.e., $\lim_{\tau \searrow 0} h_{\rec}(\tau,Q) = h_{\mathit{inv}}(Q)$.
\end{corollary}

\begin{theorem}
	\label{thm:recurrence_lessthan_tau_recurrence}
	For any $Q \subseteq \reals^n$ that is $\tau$-recurrent for some $\tau > 0$, for any $\tau' \geq \tau$, $h_{\rec}(\tau', Q) \leq h_{\rec}(\tau,Q)$ and $h_{\rec}^*(\tau', Q) \leq h_{\rec}^*(\tau,Q)$. 
\end{theorem}
\begin{proof}
	The result follows from the observation that any $(T,\varepsilon,\tau,Q)$-spanning set (resp. $(T,\tau,Q)$-spanning set) is a $(T,\varepsilon, \tau', Q)$-spanning one (resp. $(T,\tau',Q)$-spanning) as well. 
\end{proof}


	

Though Theorem \ref{thm:recurrence_lessthan_invariance} only provides a non-strict statement, it is important to notice that it only requires $Q$ to be controlled $\tau$-recurrent. As a result, it is certainly possible to have scenarios wherein
\begin{equation}
    h_{\rec}(Q,\tau)<h_{\inv}(Q)=\infty
\end{equation}
which further emphasizes the fact that achieving $\tau$-recurrence is less demanding than achieving invariance. We will show such an example in the next section.

\section{Recurrence entropy bounds}
\label{sec:bounds}

In this section, we present an upper and a lower bound on $\tau$-recurrence entropy. We show that when $\tau = 0$, we recover the upper bound on invariance entropy presented in \citep{colonius-kawan-2009}. 

\begin{theorem}[Upper bound]\label{thm:loose upper bound}
For any $\tau$-recurrent set $Q \subseteq \reals^n$ and any $\tau' \geq \tau$, $h_{\rec}(\tau', Q) \leq L_\tau \dim_F(Q) / \ln 2 \leq L_\tau n / \ln 2$, where $\dim_F(Q) := \limsup_{\varepsilon \searrow 0}\frac{\ln b(\delta,Q)}{\ln (1/\delta)}$ and $b(\delta,Q)$ is the minimal cardinality of a $\delta$-cover of $Q$.  
\end{theorem}

\begin{proof}
The proof follows that of 
Theorem 4.2 in \citep{colonius-kawan-2009}. Fix any $T$, $\varepsilon$, and $\tau' \geq \tau$. We define 
\begin{equation}\label{eq:Ltau_epsilon}
	L_{\tau, \varepsilon} =  \max_{x_1,x_2\in N_\varepsilon(\text{cl}({\mathrm{co}(R_r(Q,\tau))})),u\in U}\!\!\!\!\!\!\frac{\|f(x_1,u) - f(x_2,u)\|}{\|x_1 - x_2\|} \,<\infty.    
\end{equation}

Let $C$ be a minimal $\varepsilon e^{-L_{\tau, \varepsilon} T}$-cover of $Q$. Since $Q$ is $\tau$-recurrent, then there exists a set $S = \{u_i\}_{i \in [|C|]}$ such that $\xi(x_i, u_i, [0,T])$ is a $(T,\tau,Q)$-recurrent trajectory, where $x_i$ is the $i^{\mathit{th}}$  center  in the cover.  Using the containment lemma (i.e, Lemma~\ref{lem:containment}), we get that $\text{sup}_{t \in \reals_{\geq 0}}d(\xi(x_i, u_i, t), Q) \leq F_Q\tau e^{L_{\tau}\tau}$.


Using Grönwall's inequality,  $\forall t \in [0, T]$ and $\forall x \in B(x_i,\varepsilon e^{-L_{\tau,\varepsilon} T}) \cap Q$, $\|\xi(x_i,u_i,t) - \xi(x,u_i,t)\| \leq e^{L_{\tau, \varepsilon} t}\|x_i - x \| \leq e^{L_{\tau, \varepsilon} t}(\varepsilon e^{-L_{\tau, \varepsilon} T}) \leq \varepsilon$. Consequently, $\xi(x,u_i,t)$ is a $(T,\varepsilon,\tau,Q)$-recurrent trajectory and $S$ is a recurrence $(T,\varepsilon,\tau',Q)$-spanning set, for any $\tau' \geq \tau$. Thus, $r_{\rec}(T, \varepsilon, \tau', Q)  \leq b(\varepsilon e^{-L_{\tau, \varepsilon} T}, Q) = |S|$. Now that we have an upper bound on the minimal cardinality of a $(T,\varepsilon,\tau',Q)$-spanning set, we can get the upper bound on recurrence entropy by substituting it in equation (\ref{def:tau_rec_entropy}). Formally,
\begin{align}
	h_{\rec}(\tau',Q) &= \lim_{\varepsilon\searrow 0}\limsup_{T\rightarrow \infty} \frac{1}{T} \log r_{\rec}(T,\varepsilon, \tau', Q) \nonumber\\
	&\leq \lim_{\varepsilon\searrow 0}\limsup_{T\rightarrow \infty} \frac{1}{T} \log r_{\rec}(T,\varepsilon, \tau, Q) \nonumber\\
	&\leq \lim_{\varepsilon\searrow 0} \limsup_{T\rightarrow \infty} \frac{1}{T} \log  b(\varepsilon e^{-L_{\tau, \varepsilon} T}, Q) \nonumber \\
	&\leq  \lim_{\varepsilon\searrow 0}  \limsup_{T\rightarrow \infty} \frac{L_{\tau, \varepsilon}}{\ln (e^{L_{\tau, \varepsilon}T}/\varepsilon) + \ln \varepsilon} \log  b(\varepsilon e^{-L_{\tau, \varepsilon} T}, Q) \nonumber \\
	&=  \lim_{\varepsilon\searrow 0} L_{\tau, \varepsilon} \limsup_{T\rightarrow \infty} \frac{1}{\ln (e^{L_{\tau, \varepsilon} T}/\varepsilon)} \log  b(\varepsilon e^{-L_{\tau, \varepsilon} T}, Q) \nonumber \\
	&=   \lim_{\varepsilon\searrow 0} L_{\tau, \varepsilon} \limsup_{\delta\searrow 0} \frac{\ln  b(\delta, Q)}{\ln 2 \ln (1/\delta)}  \nonumber \\
	&= L_\tau \text{dim}_F(Q) / \ln 2.
\end{align}

The first inequality follows from the fact that any $(T, \varepsilon, \tau,Q)$-spanning set is a $(T, \varepsilon, \tau',Q)$-spanning one when $\tau \leq \tau'$. The second inequality follows from $S$ constructed earlier being a $(T, \varepsilon, \tau,Q)$-spanning set with cardinality $b(\varepsilon e^{-L_{\tau, \varepsilon} T}, Q)$. The third inequality follows from multiplying the numerator and denominator with $L_{\tau, \varepsilon}$ and using the fact that $\ln (e^{L_{\tau, \varepsilon}T}/\varepsilon) + \ln \varepsilon = L_{\tau, \varepsilon}T$. The equality after that follows from the $\limsup$ being unaffected by $\ln \varepsilon$ in the denominator and $L_{\tau, \varepsilon}$ being independent of $T$. The one before the last equality follows from replacing $\varepsilon e^{-L_{\tau, \varepsilon} T}$ with $\delta$, which transform $\limsup_{T\rightarrow \infty}$ to $\limsup_{\delta \searrow 0}$ as well as the fact that $\log c = \ln c/\ln 2$. The last equality follows from substituting the definition of $\dim_F(Q)$ and $\lim_{\varepsilon \searrow 0}L_{\tau, \varepsilon}$ by its value $L_\tau$.  
\end{proof}

\begin{remark}
Setting $\tau$ to zero makes $R_r(Q,\tau)$ as defined in (\ref{eq:R_r}) equal to $Q$ and in the definition of $L_\tau$ in (\ref{eq:Ltau}), the domain of the maximum would be cl(co($Q$)). 
Assuming $Q$ is already convex, substituting this $L_\tau$ in the bound in Theorem~\ref{thm:loose upper bound}  results in the same upper-bound as that on invariance entropy in Theorem 4.2 in \citep{colonius-kawan-2009}. 
\end{remark}

\begin{remark}
Theorem~\ref{thm:loose upper bound} shows that if the system is capable of achieving faster recurrence to $Q$ than required, i.e., achieving $\tau$-recurrence while the requirement is $\tau'$-recurrence for some $\tau' > \tau$, then we can obtain a tighter upper bound on recurrence entropy since $L_\tau \leq L_{\tau'}$. 
\end{remark}

 \begin{example}[Illustrative Example]
 	\vspace{.5ex}
 	Consider the case following two-dimensional linear system
 	\begin{align}
  \label{eq:example_system}
 		\begin{bmatrix}
 			\dot x_1 \\ \dot x_2
 		\end{bmatrix}=
 		\begin{bmatrix}
 			0 & 1\\ 0 & 0 
 		\end{bmatrix}
 		\begin{bmatrix}
 			x_1 \\ x_2
 		\end{bmatrix}+
 		\begin{bmatrix}
 			0 \\ 1
 		\end{bmatrix} u
 	\end{align}
 	We assume $u\in U=[-1,1]$, and consider the set $Q=[-1,1]^2$.

  Observe that with simple integration, we can get the closed form solution as follows: 
  \begin{equation*}
 		\xi(x,u,t)
 		=\begin{bmatrix}
 			x_1(t)\\x_2(t)
 		\end{bmatrix}
 		=\begin{bmatrix}
 			\int_{s_2 = 0}^t\int_{s_1 = 0}^t u(s_1) ds_1 d s_2 + x_2(0)t + x_1(0)\\
 			\int_{s = 0}^t u(s) ds  + x_2(0)
 		\end{bmatrix}.
 	\end{equation*}
Consider the case when $x = [1,1]$. Then, for the trajectory starting at $x$ to not leave $Q$, the control signal should be chosen so that neither of the two coordinates increase. {Both coordinates are monotonically increasing in $u$.} If we choose the control signal to have the minimum value $-1$ for some interval $[0,T]$ in the effort of preventing the state coordinates from increasing and escaping $Q$, then $\xi(
x,u,t') = [1 + t - \frac{1}{2} t^2, -t + 1]$. Thus, for all $t \leq 2$, $\xi(
x,u,t') \notin Q$. Therefore, there is no piece-wise continuous control signal that can make the trajectory starting from $x$ invariant to $Q$ or even $\tau$-recurrent with $\tau < 2$, and $Q$ is not controlled invariant or $\tau$-recurrent with any $\tau < 2$. Thus, the invariance entropy and $\tau$-recurrence entropy of system~(\ref{eq:example_system}) $h_\inv(Q)$ and $h_\rec(Q, \tau)$ are infinite for $\tau < 2$. 

In contrast, observe that with constant control signals with values in $U$, any trajectory with an initial state in $Q$ can be driven back to $Q$ within $2$ time units. Thus, $Q$ is controlled $2$-recurrent and we can use the upper bound of Theorem~\ref{thm:loose upper bound}. If we choose the $\infty$-norm, then $L_\tau \leq \|\frac{\partial f}{\partial x}\| = \|A\| = 1$. 	
  	It therefore follows that 
 	\[
 	h_{inv}(Q) = +\infty \quad\text{ and }\quad h_{rec}(Q,\tau)=
 	\begin{cases}
 		+\infty & \tau< 2\\
 		\leq 2/\ln 2 & \tau\geq 2
 	\end{cases}.
 	\]
\end{example}


\begin{theorem}[Lower bound]
\label{thm:lower_bound}
For any $\tau$-recurrent set $Q \subseteq \reals^n$, 
$$h_\rec(\tau,Q) \geq \frac{1}{\ln 2 }\max\left\{0,\min_{(x,u) \in \text{cl}(N_{\delta_\tau}(Q)) \times U} \text{div}_xf(x,u) \right\},$$ where $\text{div}_xf(x,u) = \sum_{i =1}^{n}\frac{\partial f_i}{\partial x_i}(x,u) = \text{tr}\frac{\partial f}{\partial x}(x,u)$.
\end{theorem}
\begin{proof}
A small modification of the proof of Theorem 4.1 in \citep{colonius-kawan-2009} would result in the theorem. The modified proof is as follows: first, fix $T, \varepsilon \geq 0$ and let $S = \{u_j\}_{j\in [M]}$ be a minimal  recurrence $(T, \varepsilon, \tau, Q)$-spanning set. Let us define the following sets: for any $j \in [M]$, 
\begin{align}
	Q_j = \{ x\in Q\ |\ \xi(x, u_j, [0,T]) \text{ is } (T, \varepsilon, \tau, Q)\text{-recurrent}\}.
\end{align}
For simplicity of notation, we define $\xi(Q_j, u_j,T) := \cup_{x \in Q_j}\xi(x, u_j, T)$.
Then, by the Containment lemma (Lemma~\ref{lem:containment}), {$\lambda_n(\xi(Q_j,u_j,T)) \leq \lambda_n(N_{\delta_\tau + \varepsilon}(Q))$}.
Note that in the case of invariance (as when $\tau = 0$), we instead have $\lambda_n(\xi(Q_j,u_j,T)) \leq \lambda_n(N_{\varepsilon}(Q))$, as shown in \citep{colonius-kawan-2009}. 

Now,  we can use the transformation theorem and Liouville's trace formula to get: 
\begin{align*}
	&\lambda_n(\xi(Q_j,u_j,T))= \int_{Q_j} |\text{det}\frac{\partial \xi}{\partial x} (x,u_j, T)| dx \\
	&\geq \lambda_n(Q_j) \inf_{\substack{(x,u) \in Q \times \mathcal{U}_r, \\ \xi(x,u,[0,T]) \subseteq N_{\delta_\tau + \varepsilon}(Q)}} |\text{det} \frac{\partial \xi}{\partial x}(x,u, T)| \\
	&=\lambda_n(Q_j) \inf_{\substack{(x,u) \in Q \times \mathcal{U}_r, \\ \xi(x,u,[0,T]) \subseteq N_{\delta_\tau + \varepsilon}(Q)}}\exp\big( \int_{0}^T \text{div}_x f(\xi(x,u,s),u(s)) ds\big) \\
	& \geq \lambda_n(Q_j) \min_{(x,u) \in \text{cl}(N_{\delta_\tau + \varepsilon}(Q)) \times U} \exp\big( T \text{div}_x f(x,u)\big).
\end{align*}
Now since $\lambda_n(Q) \leq M \max_{j \in [M]} \lambda_n(Q_j)$, 
\begin{align*}
	\lambda_n(Q) &\leq M \frac{\max_{j \in [M]} \lambda_n(\xi(Q_j,u_j,T))}{\min_{(x,u) \in \text{cl}(N_{\delta_\tau + \varepsilon}(Q)) \times U} \exp\big( T \text{div}_x f(x,u)\big)} \\
	&\leq M \frac{\lambda_n(N_{\delta_\tau + \varepsilon}(Q))}{\min_{(x,u) \in \text{cl}(N_{\delta_\tau + \varepsilon}(Q)) \times U} \exp\big( T \text{div}_x f(x,u)\big)}.
\end{align*}
Consequently,
\begin{align*}
	M \geq \frac{\lambda_n(Q)}{\lambda_n(N_{\delta_\tau + \varepsilon}(Q))} \min_{(x,u) \in \text{cl}(N_{\delta_\tau + \varepsilon}(Q)) \times U} \exp\big( T \text{div}_x f(x,u)\big).
\end{align*}
Recall that $M$ here is equal to $r_\rec(T,\varepsilon,\tau,Q)$. Thus, since $\tau$ is finite, $\delta_\tau$ is finite and $h_{\rec}(\tau,Q) $ 
\begin{align*}
	&\geq \lim_{\varepsilon \searrow 0} \limsup_{T\rightarrow \infty} \frac{1}{T}  \min_{(x,u) \in \text{cl}(N_{\delta_\tau + \varepsilon}(Q)) \times U} \log \exp\big( T \text{div}_x f(x,u)\big) \\
	&=  \lim_{\varepsilon \searrow 0} \frac{1}{\ln 2}\min_{(x,u) \in \text{cl}(N_{\delta_\tau + \varepsilon}(Q)) \times U} \text{div}_x f(x,u) \\
	&= \frac{1}{\ln 2}\min_{(x,u) \in \text{cl}(N_{\delta_\tau}(Q)) \times U} \text{div}_x f(x,u).
\end{align*}
Note that $\delta_\tau$ strictly increases with $\tau$. Thus, with a larger $\tau$, the domain over which the minimum is taken in the lower bound becomes larger, and the minimum itself becomes smaller. This is expected since $h(\tau',Q) \leq h(\tau, Q)$ if $\tau' \geq \tau$, according to Theorem~\ref{thm:recurrence_lessthan_tau_recurrence}. Also,  as $\tau \rightarrow 0$, we get the same lower bound as invariance entropy presented in \citep{colonius-kawan-2009}. 
\end{proof}

\begin{remark}
Theorem~\ref{thm:lower_bound} does not follow directly from Theorem 4.1 in \citep{colonius-kawan-2009}, i.e., from the result that $h_{\inv}(Q)$ is greater than or equal to $\max\{0,\min_{(x,u) \in Q \times U}\text{div}_xf(x,u)\}$, since $h_{\rec}(\tau, Q) \leq h_{\inv}(Q)$, for any $\tau \geq 0$.  
\end{remark}

\section{Entropy and $\tau$-Recurrence data rates}\label{sec:data rates}
We assume the setup where there is a sensor that can accurately measure the state of system~(\ref{eq:system}) at any time instant. It also has computation capabilities that allows it to simulate the system starting from any initial state and following any control, as long as that trajectory exists. The sensor is connected to a controller over a limited-bandwidth channel. The controller does not have information about the state of the system besides what it receives from the sensor. It does however know the $\tau$-recurrent set $Q \subseteq \reals^n$, the corresponding control signal that drives the system when starting from any state in $Q$ to have an $( \varepsilon,\tau, Q)$-recurrent trajectory.

An $(\varepsilon,\tau,Q)$-{\em recurrence enforcing} algorithm is a pair of procedures, one for the sensor and the other for the controller. The sensor's procedure determines the bits it sends over the channel to the controller. Based on these bits, the controller's procedure determines how to map these bits to a control signal to drive the system to have an $( \varepsilon,\tau, Q)$-recurrent trajectory.  The average bit rate of an $(\varepsilon,\tau,Q)$-recurrence enforcing algorithm is defined as follows: $\lim_{T\rightarrow \infty} \frac{\# \text{bits}(T)}{T}$, where $\# \text{bits}(T)$ is the total number of bits sent by the sensor until time $T$. 

\begin{theorem}
For any controlled $\tau$-recurrent set $Q \subseteq \reals^n$ and $\varepsilon \geq 0$, there exists no  $(\varepsilon,\tau,Q)$-recurrence enforcing algorithm with an average bit rate smaller than $h_\rec(\tau,Q)$. 
\end{theorem}
\begin{proof}
The proof is by contradiction. If there is such an algorithm with an average data rate smaller than entropy, then there exists a $T > 0$ such that 
\begin{align}
	\frac{\# \text{bits}(T)}{T} &< \frac{1}{T} \log r_{\rec}(T,\varepsilon,\tau,Q).
\end{align}
That implies that $2^{\# \text{bits}(T)} < r_{\rec}(T,\varepsilon,\tau,Q)$. Observe that $2^{\# \text{bits}(T)} $ is the number of control signals that the controller can possibly generate over the interval $[0,T]$. By the assumption that the controller enforces the system to have an $(T,\varepsilon,\tau,Q)$-recurrent trajectory, then for every $x \in Q$, it can generate a control signal that results in a $(T,\varepsilon, \tau, Q)$-recurrent trajectory. Therefore, the set of control signals that the controller can generate is a $(T,\varepsilon,\tau, Q)$-spanning one that has a smaller cardinality than $r_{\rec}(T,\varepsilon, \tau, Q)$, which contradicts the latter's definition being the minimal cardinality of a $(T,\varepsilon,\tau, Q)$-spanning set. 
\end{proof}

\section{Algorithm for enforcing $\tau$-recurrence over limited-bandwidth channels}\label{sec:algorithm}

In this section, we present Algorithm~\ref{code:recurrence}, which when run at the sensor and a corresponding procedure running at the controller, it can produce a control signal for system~(\ref{eq:system}) that drives its trajectory to be exponentially converging to a $(\tau, Q)$-recurrent one at a user-specified rate $\alpha \geq 0$. When $\alpha = 0$, the trajectory would be an $(\varepsilon,\tau,Q)$-recurrent trajectory with a user-specified $\varepsilon$. We define this more formally in Theorem~\ref{thm:algorithm_guarantees} and Corollary~\ref{cor:asymptotic_recurrence}. After that, we show that the bit rate at which the sensor should send information to the controller is equal to the upper bound on $\tau$-recurrence entropy presented in Theorem~\ref{thm:loose upper bound} when $\alpha = 0$, and grows linearly with $\alpha$, otherwise. 

In our algorithm, we assume that starting from any state in $N_{\delta_\tau + \varepsilon}(Q)$, for some $\varepsilon > 0$, there exists a control signal that drives system~(\ref{eq:system}) to $Q$ within $\tau$ time units. Moreover, we assume that the function 
that maps the initial states $N_{\delta_\tau + \varepsilon}(Q)$ to the shortest time such a control signal takes to drive system~(\ref{eq:system})  to $Q$ to be Lipschitz continuous. This is formulated as follows. 

\begin{assumption}
\label{ass:can_return_to_Q}
$\exists \varepsilon^* > 0$ such that there exists a control function $h: N_{\delta_\tau + \varepsilon^*}(Q) \times \reals^{\geq 0} \rightarrow U$ and a corresponding function $\mathit{ttq}: N_{\delta_\tau + \varepsilon^*}(Q) \rightarrow [0,\tau]$, such that 
$\forall x \in N_{\delta_\tau + \varepsilon^*}(Q)$, 
$\xi(x, h(x,\cdot), \mathit{ttq}(x)) \in Q$ and $\forall t \in [0,\mathit{ttq}(x)]$, $d(\xi(x, h(x,\cdot), t), Q) \leq d(x, Q)$. 
Moreover,  there exists some constant $c^* \geq 0$ such that for any $x_1, x_2 \in N_{\delta_\tau + \varepsilon^*}(Q)$,
$|\mathit{ttq}(x_1) - \mathit{ttq}(x_2)| \leq c^* \|x_1 - x_2\|$.  
\end{assumption}

Next, for any $\tau > 0$ and $\varepsilon \in (0,\varepsilon^*]$, we define a new control function $g: N_{\delta_\tau + \varepsilon}(Q) \times \reals^{\geq 0}\rightarrow U$ to be used in the algorithm. If the initial state $x$ is in $Q$, $g(x,\cdot)$ is equal to a control signal that ensures $(\tau,Q)$-recurrence, which exists by the assumption that $Q$ is  $\tau$-recurrent. Otherwise, it is equal to the control function $h$ defined in Assumption~\ref{ass:can_return_to_Q} 
up until reaching $Q$, i.e., until $\mathit{ttq}(x)$. After that, it is equal to the control function that ensures the trajectory is $\tau$-recurrent starting from the new initial state in $Q$.  

Formally, let $g': Q \times \reals^{\geq 0} \rightarrow U$ be such that for any $x\in Q$, the trajectory  $\xi(x,g'(x,\cdot), \cdot)$
is a $(\tau, Q)$-recurrent one. Such a function exists because of the assumption that $Q$ is a  $\tau$-recurrent set. 
We define $g$ as follows: $\forall t \geq 0, g(x,t) := g'(x,t)$ if $x \in Q$. If $x \in  N_{\delta_\tau + \varepsilon^*}(Q) \text{\textbackslash} Q$, $\forall t \leq \mathit{ttq}(x), g(x,t) := h(x,t)$ and $\forall t > \mathit{ttq}(x), g(x,t) := g'(\xi(x,h(x,\cdot),\mathit{ttq}(x)), t - \mathit{ttq}(x))$.

\subsection{Algorithm description}

Algorithm~\ref{code:recurrence} takes as input a $\tau$-recurrent set $Q$ for some $\tau > 0$, an $\varepsilon \in (0,\varepsilon^*]$ (where $\varepsilon^*$ is as defined in Assumption~\ref{ass:can_return_to_Q}), and the control function $g$ defined earlier. It also assumes to be given several functions:  $\mathit{sense}$,  $\mathit{quantize}$, $\mathit{encode}$, $\mathit{send}$, $\mathit{simulate}$, and $\mathit{sleep}$. The function $\mathit{sense}$ returns the current state of the system. The function  $\mathit{quantize}$ returns the closest point in the set given in its second argument to the point given in its first argument, according to $\infty$-norm. The function $\mathit{encode}$ maps the first argument to a bit vector that uniquely identifies it out of the set of states given in the second argument. The function $\mathit{send}$ sends the given bit vector over the limited bandwidth channel to the controller.  The function $\mathit{simulate}$ simulates the system starting from the state in its first argument following the control signal in its second argument until the time bound specified in its third argument. It returns the last state in the simulated trajectory. If the third argument is an interval, it returns the trajectory segment within that interval. Finally, the function $\mathit{sleep}$ makes the sensor wait for the amount of real time passed as argument before continuing the execution of the algorithm. The time of the algorithm execution is assumed to be negligible with respect to $\tau$.

The algorithm starts by initializing $S_0$ to $Q$ and constructing an $\varepsilon e^{-(L_\tau + \alpha)\tau}$-grid for it, which we denote by $C_0$. The algorithm then proceeds with an infinite loop. In each iteration, it sends a bit vector that encodes a state estimate to the controller, according to which it can identify the control function $u_i$ the system should follow in the time interval $[i\tau, (i+1)\tau)$. To produce the bit vector, the sensor measures the current state of the system $x_i$, i.e., $ \xi(x_0, \hat{u}, i\tau)$, where $\hat{u}$ is the control signal have been followed so far. Then, it quantizes $x_i$  to one of the centers $q_i$ in the grid $C_i$. The encoding of $q_i$ with respect of $C_i$ is the bit vector that the sensor sends. The controller, which is running a similar algorithm to Algorithm~\ref{code:recurrence}, but without the sensing, can recover $q_i$ as it knows $C_i$. Using $q_i$, it can choose the same control function $u_i$ that the sensor intends to use to construct $C_{i+1}$. 

After that, Algorithm~\ref{code:recurrence} computes $u_i$ to be equal to $g(q_i,[0,\tau))$.  
Then, it simulates the system for $\tau$ time units starting from $q_i$ and following $u_i$. It uses the last state in the simulated trajectory as the center of the ball $S_{i+1}$ which bounds the region where the next sensed state $x_{i+1}$ might be. The radius $r_{i+1}$ of $S_{i+1}$ is an $e^{\alpha \tau}$ factor smaller than that of $S_i$. After that, it constructs the grid $C_{i+1}$ to be the $r_{i+1}e^{-(L_\tau + \alpha)\tau}$-grid over $S_{i+1}$, according to which the next state, $x_{i+1}$, would be quantized. Finally, the sensor waits for the system to evolve for $\tau$ time units before sensing it again in the next iteration. 

\begin{algorithm}
\caption{Sensor algorithm for achieving recurrence} 
\label{code:recurrence}
\begin{algorithmic}[1]
\State {\bf input:} $Q$, $\varepsilon \in (0, \varepsilon^*]$, $\tau > 0$, $g: N_{\delta_\tau + \varepsilon}(Q) \times \reals^{\geq 0}\rightarrow U$
\State $S_0 \gets Q$
\State $r_0 \gets \varepsilon$
\State $C_0 \gets \mathit{grid}(S_0, r_0 e^{-(L_\tau+ \alpha) \tau})$\label{ln:initial_grid}
\item $i = 0$
\While{true}
\State $x_i \gets \mathit{sense}()$
\State $q_i \gets \mathit{quantize}(x_i, C_i)$
\State $\mathit{send}(\mathit{encode}(q_i, C_i))$
\State $u_i \gets g(q_i, [0,\tau))$
\State $r_{i+1} \gets r_{i} e^{-\alpha\tau}$
\State $S_{i+1} \gets B(\mathit{simulate}(q_{i},u_i,\tau), r_{i+1})$
\State $C_{i+1} \gets \mathit{grid}(S_{i+1}, r_{i+1}e^{-(L_\tau+ \alpha) \tau})$
\State $i \gets i + 1$
\State $\mathit{sleep}(\tau)$
\EndWhile
\end{algorithmic}
\end{algorithm}

\subsection{Algorithm guarantees}

Fix the inputs to Algorithm~\ref{code:recurrence}, i.e., a controlled $\tau$-recurrent set $Q$ and a corresponding $\tau$-recurrence achieving controller $g$. Moreover, fix any initial state $x_0\in Q$. Let $\hat{u}: \reals^{\geq 0} \rightarrow \reals^m$ be the concatenation of the $u_i$s produced by Algorithm~\ref{code:recurrence}, i.e., for any $t \geq 0$, $\hat{u}(t) =  u_i(t - i\tau)$, where $i = \lfloor t/\tau \rfloor$. Also, let $\hat{\xi}: \reals^{\geq 0} \rightarrow \reals^m$ be the concatenation of the $\tau$-sized fragments of trajectories $\xi(q_i, u_i, [0, \tau))$ produced by the algorithm, i.e., $\forall t \geq 0$ and $t \neq i\tau$ for some $i \in \mathbb{N}$, $\hat{\xi}(t) = \mathit{simulate}(q_i, u_i, t - i\tau)$ and $\hat{\xi}(i\tau) = q_i$, where $i = \lfloor t/\tau \rfloor$. Thus, $\hat{\xi}$ would be right-piece-wise-continuous. Finally, the trajectory that the system would have starting from $x_0$ following $\hat{u}$ is denoted as usual by $\xi(x_0,\hat{u}, \cdot)$.

\begin{theorem}
\label{thm:algorithm_guarantees}
Algorithm~\ref{code:recurrence} ensures that:
\begin{enumerate}
\item $\forall i \geq 0$, $x_i \in S_i$, and $\forall  t\geq 0, \|\hat{\xi}(t) - \xi(x_0, \hat{u}, t)\| \leq \varepsilon e^{-\alpha t}$,\label{thm:distance_between_xi_and_xi_hat}
\item $\forall i \in \mathbb{N}$, $\hat{\xi}[i\tau, \infty)$ is an $(\varepsilon e^{-i\alpha\tau}, \tau + c^* \varepsilon e^{-(i\alpha + L_\tau) \tau}, Q)$-recurrent function, and\label{thm:xi_hat_is_recurrent}
\item $\forall i \in \mathbb{N}$, $\xi(x_0, \hat{u}, [i\tau, \infty))$ is a $(2\varepsilon e^{-i\alpha \tau}, \tau + c^* \varepsilon e^{-(i\alpha + L_\tau) \tau}, Q)$-recurrent trajectory.\label{thm:xi_is_recurrent}  
\end{enumerate}

\end{theorem}

\begin{proof}
First, we will prove part \ref{thm:distance_between_xi_and_xi_hat}) by induction. For the base case: $x_0 \in S_0$ and $\|\hat{\xi}(0)  - x_0 \|\leq \varepsilon e^{-(L_\tau + \alpha) \tau} \leq \varepsilon$, which hold by the fact that $C_0$ is a grid over $Q$ with cells of radii $r_0 = \varepsilon e^{-(L_\tau + \alpha) \tau}$ and $\hat{\xi}(0) = q_0$.

Inductive case: fix an $i \in \mathbb{N}$ and assume that 
$x_i \in S_i$ and $\forall t \in [0, i\tau]$, $\|\hat{\xi}(t) - \xi(x_0, \hat{u}, t)\| \leq \varepsilon e^{-\alpha t}$.   
By Grönwall's inequality, 
$\forall t\in [i\tau,(i+1) \tau], \|\xi(x_i,u_i,t) - \xi(q_i,u_i,t)\| \leq e^{L_\tau (t - i\tau)}\|x_i - q_i\| \leq \varepsilon e^{- ((i+1)\alpha + L_\tau)\tau } e^{L_\tau (t - i\tau)} \leq \varepsilon e^{-(i+1)\alpha \tau} \leq \varepsilon e^{-\alpha t}$.
Recall that $x_{i+1} = \xi(x_i,u_i,\tau)$ and $r_{i+1} = \varepsilon e^{-(i+1)\alpha \tau}$.
Thus, $x_{i+1} \in B(\xi(q_i,u_i,\tau), r_{i+1})$, and the latter is $S_{i+1}$. Then, since $C_{i+1}$ is a grid over $S_{i+1}$ with granularity $r_{i+1} e^{- (L_\tau+ \alpha) \tau}$ and $ \hat{\xi}((i+1)\tau) = q_{i+1}$, $\|x_{i+1} - \hat{\xi}((i+1)\tau)\| = \|x_{i+1} - q_{i+1}\|\leq r_{i+1} e^{-(L_\tau+ \alpha) \tau}$, and thus $\|\xi(x_0, \hat{u}, (i+1)\tau) - \hat{\xi}((i+1)\tau)\|\leq \varepsilon e^{- ((i+2)\alpha + L_\tau) \tau} \leq \varepsilon e^{-(i+1)\alpha \tau}$. 
That proves the inductive argument for part \ref{thm:distance_between_xi_and_xi_hat}). 

We prove part \ref{thm:xi_hat_is_recurrent}) also by induction. We will prove the stronger claim that for any $i \in \mathbb{N}$, either $\hat{\xi}(t) \in Q$ for some $t\in [i\tau, (i+1)\tau)$ or $\lim_{t\rightarrow ((i+1)\tau)^-} \hat{\xi}(t) \in Q$, and $|t_{i+1} - t_{i}| \leq \tau + c^* \varepsilon e^{-i\alpha \tau}$, where $t_{i}$ is the last time instant in $[i\tau, (i+1)\tau)$ such that $\hat{\xi}(t_{i}) \in Q$ or equal to $ (i+1)\tau$, otherwise, and $t_{i+1}$ is the first time instant in $[(i+1)\tau, (i+2)\tau)$  such that $\hat{\xi}(t_{i+1}) \in Q$, or $t_{i+1}^* = (i+2)\tau$, otherwise.  When it is the case that $\lim_{t\rightarrow ((i+1)\tau)^-} \hat{\xi}(t) \in Q$, we know from part~\ref{thm:distance_between_xi_and_xi_hat}) that $q_{i+1} \in S_{i+1}$, which is centered at the value of that limit and has a radius of $r_{i+1}$. Thus, $\hat{\xi}((i+1)\tau)$, which is equal to $q_{i+1}$, would be at most $r_{i+1}$ (i.e., $\varepsilon e^{-(i+1)\alpha \tau}$) from $Q$. 

Base case: by assumption, $x_0 \in Q$ and $\hat{\xi}(0) = q_0$. By part~\ref{thm:distance_between_xi_and_xi_hat}), $\| \hat{\xi}(0) - x_0\| \leq \varepsilon e^{-(L_\tau + \alpha)\tau}$ and thus $\hat{\xi}(0) \in N_{\varepsilon e^{-(L_\tau + \alpha)\tau}}(Q) \subseteq N_{\delta_\tau + \varepsilon e^{-(L_\tau + \alpha)\tau}}(Q)$. If $q_0 \in Q$, then $u_0$ is equal to $g(q_0, [0,\tau))$. That would result in $\hat{\xi}([0, \tau))$ being a prefix of a $(\tau, Q)$-recurrent trajectory starting from $q_0$, by the definition of $g$. Thus, either $\hat{\xi}(t) \in Q$ for some $t \in (0,\tau)$ or $\lim_{t \rightarrow \tau^-}\hat{\xi}(t) \in Q$. 
%
Moreover, by the containment lemma (Lemma~\ref{lem:containment}), $\forall t\in [0, \tau), \hat{\xi}(t) \in N_{\delta_\tau}(Q)$. 
Thus, $\lim_{t \rightarrow \tau^-}\hat{\xi}(t)$, which is the center of $S_1$, would be in $N_{\delta_\tau}(Q)$. Thus, $\hat{\xi}(\tau)$, which is equal to $q_1$ and belongs to $S_1$, 
would be in $N_{\delta_\tau + \varepsilon e^{-\alpha \tau}}(Q)$.  

If $\hat{\xi}(0) \in N_{\delta_\tau + \varepsilon}(Q)$\textbackslash$Q$ instead, then, by Assumption~\ref{ass:can_return_to_Q}, applying the control $g(q_0, [0, \tau])$  will result in $\hat{\xi}(t) \in N_{\delta_\tau + \varepsilon}(Q)$, for all $t \in [0, \mathit{ttq}(x)]$.
If $\mathit{ttq}(x) = \tau$, then as in the first case, $\lim_{t\rightarrow \tau^-}\hat{\xi}(t) \in Q$
Otherwise, if $\mathit{ttq}(x) < \tau$, then $\hat{\xi}(\mathit{ttq}(x)) \in Q$. Moreover, in the interval $(\mathit{ttq}(x),\tau)$, $\hat{\xi}$ 
would be equal to the trajectory $\xi(\xi(q_0, u_0, \mathit{ttq}(x)),$ $u_0((\mathit{ttq}(x), \tau)), \cdot)$, which is a $(\tau,Q)$-recurrent trajectory by the definition of $u_0((\mathit{ttq}(x),\tau])$ being equal to $g(\xi(q_0, u_0, \mathit{ttq}(x)), [t-\mathit{ttq}(x), \tau - \mathit{ttq}(x)))$. Thus, by the containment lemma (Lemma~\ref{lem:containment}), it is contained in $N_{\delta_\tau}(Q)$. 
Hence, $\lim_{t\rightarrow \tau^-} \hat{\xi}(t)$, the center of $S_1$, is in $N_{\delta_\tau}(Q)$.
By part \ref{thm:distance_between_xi_and_xi_hat}), $x_1 \in S_1$. Also, $q_1 \in S_1$ and $\hat{\xi}(\tau) = q_1$. Thus, $\hat{\xi}(\tau) \in N_{\delta_\tau + r_1}(Q) = N_{\delta_\tau + \varepsilon e^{-\alpha \tau}}(Q)$.

Inductive case: fix an $i \geq 1$ and assume that that part~\ref{thm:xi_hat_is_recurrent}) is true until time $i\tau$. 
Thus, there exists a time instant $t \in [0,\tau)$ such that $\hat{\xi}((i-1)\tau + t) \in Q$ or $\lim_{t\rightarrow (i\tau)^-} \hat{\xi}(t) \in Q$. 
Let $t_{i-1}$ be the largest such instant.
If $t_{i-1} <\tau$ and we simulate system~(\ref{eq:system}) following $g(\xi(q_{i-1},u_{i-1}, i\tau + t_{i-1}), \cdot)$ starting from $\xi(q_{i-1},u_{i-1}, i\tau + t_{i-1})$, the resulting trajectory will be $(\tau,Q)$-recurrent. Thus, there exists $t' \in (0,\tau]$ such that that trajectory belongs to $Q$ at time $i\tau + t'$. However, $\hat{\xi}$ is equal to that trajectory only in the interval $[(i-1)\tau + t_{i-1}, i\tau)$. In the interval $[i\tau, (i+1)\tau)$, $\hat{\xi}$ will be equal to the trajectory that starts from $q_i$ and follows $g(q_i, \cdot)$. If $q_i \in Q$, then $\hat{\xi}$ would have visited $Q$ within $\tau - t_{i-1}^{*}$ time units, which is less than $\tau$. If $q_i \notin Q$, then from part~\ref{thm:distance_between_xi_and_xi_hat}), we know that $q_i \in S_i$ and thus $\|q_i - \lim_{t\rightarrow i\tau^-}\hat{\xi}(t)\|\leq r_i = \varepsilon e^{- i \alpha \tau}$. Then, by Assumption~\ref{ass:can_return_to_Q}, we know that $\hat{
\xi}$ would reach $Q$ at or before $\min \{(i+1)\tau, t_{i-1} + \tau + c^* \|q_i -  \lim_{t\rightarrow i\tau^-}\hat{\xi}(t) \|\}$, which is upper bounded by  $\min \{(i+1)\tau, t_{i-1} + \tau + c^*\varepsilon e^{- i \alpha \tau}\}$. We can conclude that the time between two time instants at which $\hat{\xi}$ belongs to $Q$ in the intervals $[(i-1)\tau, i\tau)$ and $[i\tau, (i+1)\tau)$ is less than or equal to $\tau +  c^*\varepsilon e^{- i \alpha \tau}$. 


Finally, part~\ref{thm:xi_is_recurrent}) follows from combining parts~\ref{thm:distance_between_xi_and_xi_hat}) and  part~\ref{thm:xi_hat_is_recurrent}) and using the triangular inequality $d(\xi(x_0,\hat{u},t), Q) \leq \|\xi(x_0,\hat{u},t) - \hat{\xi}(t)\| + d(\hat{\xi}(t), Q)$ at the time instants where $\hat{\xi}(t)$ is visiting $N_{\delta_\tau + \varepsilon e^{-i\alpha\tau}}(Q)$. We obtain that $\xi(x_0,\hat{u},\cdot)$ visits $N_{\delta_\tau + 2\varepsilon e^{-i\alpha \tau}}(Q)$ in the $[i\tau, (i+1)\tau)$ interval. Thus, $\xi(x_0,\hat{u}, [i\tau, \infty))$ is a $(2\varepsilon e^{-i\alpha \tau}, \tau +  c^*\varepsilon e^{- i \alpha \tau}, Q)$-recurrent trajectory.
\end{proof}

It follows that the trajectory of system~(\ref{eq:system}) when following the controller $\hat{u}$ produced by Algorithm~\ref{code:recurrence} asymptotically approaches a $(\tau, Q)$-recurrent trajectory.
\begin{corollary}
\label{cor:asymptotic_recurrence}
As $t\rightarrow \infty$, $\xi(x_0, \hat{u}, [t, \infty))$ is a $(\tau, Q)$-recurrent trajectory.
\end{corollary}



In the following theorem, we show that the bit rate of Algorithm~\ref{code:recurrence} matches the upper bound on $\tau$-recurrence entropy we presented in Section~\ref{sec:bounds}.
\begin{theorem}
\label{thm:algorithm_bitrate}
The average bit rate at which a sensor running  Algorithm~\ref{code:recurrence} will send to the controller is equal to $n(L_\tau + \alpha) / \ln 2$. 
\end{theorem}
\begin{proof}
Fix any $i\in \mathbb{N}$. The number of bits that the sensor running Algorithm~\ref{code:recurrence} sends at the time instant $t = i\tau$ is $\log |C_i|$. Given any time bound $T \geq 0$, the total number of bits sent by the sensor over $[0,T]$ is equal to $\sum_{i = 0}^{\lfloor T/\tau\rfloor} \log |C_i|$. Thus, the average bit rate is $\lim_{T \rightarrow \infty} \sum_{i = 0}^{\lfloor T/\tau\rfloor} \frac{\log |C_i|}{T}$.

We can observe that $\forall i \geq 0$, $C_i = \lceil\frac{\text{diam}(S_i)}{2 r_i e^{-(L_\tau + \alpha)\tau}}\rceil^n$. Then, $C_0 = \lceil\frac{\text{diam}(Q)}{2 \varepsilon e^{-(L_\tau + \alpha)\tau}}\rceil^n$ and for any $i \geq 1$, $C_i = \lceil \frac{2r_i}{2 r_i e^{-(L_\tau + \alpha)\tau}} \rceil^n = \lceil e^{(L_\tau + \alpha)\tau} \rceil^n$. 

Thus, the average bit rate is equal to:
\begin{align}
& \lim_{T\rightarrow \infty} \frac{1}{T} \sum_{i = 0}^{\lfloor T/\tau\rfloor} \frac{\log |C_i|}{T} \nonumber \\ 
&= \lim_{T\rightarrow \infty} \frac{1}{T} \big( \lceil\frac{\text{diam}(Q)}{2 \varepsilon e^{-(L_\tau + \alpha)\tau}}\rceil^n + \sum_{i = 1}^{\lfloor T/\tau\rfloor} \log  \lceil e^{(L_\tau + \alpha)\tau} \rceil^n\big) \nonumber \\
&= n (L_\tau + \alpha)/\ln 2.
\end{align}

\end{proof}



\section{Conclusions and Future Work}\label{sec:conclusions}
We present the notion of $\tau$-recurrence entropy for nonlinear control systems as a generalization of the notion of invariance entropy. In $\tau$-recurrence of the system with respect to a predefined compact set, the trajectories can leave it, but only for $\tau$ time units whenever it does. $\tau$-recurrence entropy measures the exponential rate at which the number of control signals that are sufficient to make the system $\tau$-recurrent increases with time. We show that $\tau$-recurrence entropy is bounded from above and below by the invariance entropy of the system with respect to different compact sets. Moreover, we show that it converges to invariance entropy with respect to the same set as $\tau$ decreases, as expected. Then, we derive upper and lower bounds on $\tau$-recurrence entropy as a function of the system dimension, local Lipschitz constant, and the divergence of the vector field. We show that both bounds converge to known corresponding bounds on invariance entropy as $\tau \rightarrow 0$, as expected. We then show that the average bit rate of a recurrence-achieving algorithm
is lower bounded by the $\tau$-recurrence entropy. Finally, we present such an algorithm 
that guarantees exponential convergence to a $(\tau,Q)$-recurrent one with an average bit rate equal to the upper bound on entropy we derived plus a linear term in the rate of convergence multiplied by the state dimension. 

A possible future direction would be to design an algorithm similar to Algorithm~\ref{code:recurrence} that instead of constructing a moving grid online each $\tau$ seconds, it uses fixed grid. This might increase the required bit rate, but would save the sensor and controller from significant online computations that can be done offline instead. Another direction would be to design numerical methods to estimate entropy using abstractions, as \citep{tomar2022numerical}.

\bibliographystyle{IEEEtran}
\bibliography{refs,mallada,hussein_entropy}

\end{document}